\documentclass[prb,twocolumn]{revtex4-1} 

\usepackage{bm}
\usepackage{tikz-cd}
\usepackage[ruled,vlined,algosection,norelsize,algo2e]{algorithm2e}
\usepackage{mathptmx}
\usepackage{amssymb}
\usepackage{amsthm}
\usepackage{amsmath}    
\usepackage{amsfonts} 
\usepackage{mathtools}
\usepackage{graphicx}   
\newcommand{\ket}[1]{\left|#1\right\rangle}
\newcommand{\bra}[1]{\left\langle#1\right|}
\newcommand{\overlap}[2]{\left\langle#1|#2\right\rangle}

\renewcommand*{\mathellipsis}{%
  \mathinner{{\ldotp}{\ldotp}{\ldotp}}%
}
\usepackage{paralist}

\newtheorem{lemma}[algocf]{Lemma}
\newtheorem{corollary}[algocf]{Corollary}
\newtheorem{theorem}[algocf]{Theorem}

\theoremstyle{definition}

\numberwithin{algocf}{section}

\begin{document}
\title{Multiple photon effective Hamiltonian in linear quantum optical networks}

\author{Juan Carlos Garcia-Escartin}
\email{juagar@tel.uva.es} 
\affiliation{Universidad de Valladolid, Dpto. Teor\'ia de la Se\~{n}al e Ing. Telem\'atica, Paseo Bel\'en n$^{\circ}$ 15, 47011 Valladolid, Spain.}
\author{Vicent Gimeno}
\email{gimenov@uji.es} 
\affiliation{Universitat Jaume I, Departamento de Matem\'aticas and IMAC-Institut Universitari de Matem\`atiques i Aplicacions de Castell\'o, 12071 Castell\'on de la Plana, Spain.}
\author{Julio Jos\'e Moyano-Fern\'andez}
\email{moyano@uji.es} 
\thanks{All authors contributed equally to this work.}
\affiliation{Universitat Jaume I, Departamento de Matem\'aticas and IMAC-Institut Universitari de Matem\`atiques i Aplicacions de Castell\'o, 12071 Castell\'on de la Plana, Spain.}


\date{\today}

\begin{abstract}
We give an alternative derivation for the explicit formula of the effective Hamiltonian describing the evolution of the quantum state of any number of photons entering a linear optics multiport. The description is based on the effective Hamiltonian of the optical system for a single photon and comes from relating the evolution in the Lie group that describes the unitary evolution matrices in the Hilbert space of the photon states to the evolution in the Lie algebra of the Hamiltonians for one and multiple photons. We give a few examples of how a group theory approach can shed light on some properties of devices with two input ports.
\end{abstract}


\maketitle 

\section{Linear quantum-optical networks}
The evolution of the quantum states of light when they pass linear optical networks can be described from the classical scattering matrix of the network $S$. In classical electromagnetism, $S$ relates the amplitudes of the input fields in $m$ input modes with the amplitudes of the $m$ output modes and has many applications in microwave circuit design\cite{Poz04}.

In quantum optics, we can replace the field amplitudes with the probability amplitudes in the wavefunction of a photon and use $S$ to see the evolution of the creation operator in each mode. 

However, when there are multiple photons, the evolution does not only include wave interference effects, like in classical electromagnetic waves, but also purely quantum effects related to the bosonic nature of the photons. A most striking example is the Hong-Ou-Mandel effect in which two independent photons that reach simultaneously the two separate inputs of a beam splitter always come out together\cite{HOM87}. These interactions have no classical counterpart and are behind the ability of linear optical systems to give an efficient solution to the boson sampling problem, a task that is strongly believed to be inefficient in any classical computing machine\cite{AA11}. 

In this paper, we describe the quantum evolution of photons through linear optical elements from results from group theory. We work with photonic states
\begin{equation}
\label{basisstates}
\bigotimes_{k=1}^{m}\ket{n_k}_k=\ket{n_1}_1\ket{n_2}_2\mathellipsis\ket{n_m}_m=\ket{n_1n_2\ldots n_m}
\end{equation}
with $n$ photons that are distributed into $m$ orthogonal modes. In the most general case, these modes represent any set of orthogonal single photon states so that
\begin{equation}
_k\overlap{1}{1}_l=\delta_{k,l}.
\end{equation}
The modes can be different paths, which gives a very intuitive picture of the network, but they can also represent orthogonal temporal wavefunctions, different directions in the same spatial path, photons in orthogonal polarization states, photons with different orbital angular momentum or with a different frequency. 

We consider linear optical systems where the number of photons is conserved
\begin{equation}
\sum_{k=1}^{m}n_k=n.
\end{equation}
In passive lossless systems the total energy is conserved, which fits well with the description in terms of classical fields. The quantum equivalent is a conservation of probability. If we have a superposition of $n$ photon states, the output will be a different superposition where all the photon states must sum to a probability of one. The input and output states are related by a unitary operator
\begin{equation}
\ket{\psi}_{out}=U\ket{\psi}_{in}.
\end{equation}

The same mathematical description could be extended to active linear systems as long as the number of photons is preserved. In principle, it could include elements where a change of frequency introduces an energy change, provided that the total number of photons in the $m$ modes of interest does not change.

In quantum optics and quantum information we usually find states living in a finite-dimensional Hilbert space and the operators can be written as matrices. Our states live in a Hilbert space $\mathcal{H}$ of a size $M={m+n-1 \choose n}$. A generic state $\ket{\psi}$ can be described as a linear combination of the basis elements from Equation (\ref{basisstates}) that exhaust all the possible ways to distribute $n$ photons into $m$ modes. The problem is equivalent to counting the number of ways to place $n$ balls into $m$ boxes.

\subsection{Unitary evolution}
In this finite-dimensional Hilbert space, we can write the states as complex column vectors and the unitary operators $U$ as $M\times M$ unitary matrices. For systems with exactly one photon $S=U$ and the quantum state of the photon in mode $k$, $\ket{1}_k$, is represented as a column vector with $m$ rows that are zero except for the $k$th row, which has a one.

For $n$ photons there is a known, more involved transformation that gives $U$ in terms of $S$. The evolution matrix $U$ of the system with $n$ photons belongs to the unitary group $U(M)$ and $S\in U(m)$. We can define a group homomorphism $\varphi: U(m)\to U(M)$.

A nice description of the properties of $\varphi$ and a verification that it is indeed a homomorphism between groups is given in Aaronson and Arkhipov's paper\cite{AA11}. We content ourselves with noticing the physically relevant fact that the homomorphism preserves the group operation to compose operators, which is matrix multiplication in our description. If we have a succession of $N$ optical networks, the first with a matrix $S_1$ and the last with $S_N$, the total system has a scattering matrix $S=S_N \ldots S_2S_1$. Their effect of $n$ photons can also be described by multiplying the corresponding unitary operators with $U(S)=U(S_N)\ldots U(S_2)U(S_1)$, as expected.

Here and in the following sections, we work with the usual photon creation and annihilation operators $\hat{a}_k^{\dagger}$ and $\hat{a}_k$\cite{Lou00}. Their effect on states with $n_k$ photons in mode $k$ is
\begin{align}
\label{CreatAnni}
\hat{a}_k^{\dagger}\ket{n_k}_k=\sqrt{n_k+1}\ket{n_k+1}_k&,&  \\
\hat{a}_k\ket{n_k}_k=\sqrt{n_k}\ket{n_k-1}_k, \hspace{2ex} n \geq 1&,&\hspace{1ex} \hat{a}_k\ket{0}_k=0.
\end{align}

There are a few equivalent ways to write $\varphi$. For our purposes, we prefer the description in terms of the evolution of the operators in the Heisenberg picture, which shows how all the operators $a_k^{\dag}$ evolve for all the indices from 1 to $m$\cite{SGL04}. For any $n$-photon input state
\begin{equation}
\ket{n_1 n_2 \ldots n_m}=\prod_{k=1}^{m} \left(\frac{\hat{a}_k^{\dag n_k}}{\sqrt{n_k!}}\right)\ket{00\ldots 0}
\end{equation}
the output state is given from the elements of $S$ as
\begin{equation}
\label{HeisenbergU}
U\ket{n_1 n_2 \ldots n_m}=\prod_{k=1}^{m} \frac{1}{\sqrt{n_k!}}\left(\sum_{j=1}^{m}S_{jk} \hat{a}_j^{\dag}\right )^{n_k}\ket{00\mathellipsis0}.
\end{equation}

Each element of $U$ can be deduced from Equation (\ref{HeisenbergU}). For an input state $\ket{n_1 n_2 \ldots n_m}$ and an output $\ket{n_1' n_2' \ldots n_m'}$, $\bra{n_1'n_2'\ldots n_m'}U\ket{n_1 n_2 \ldots n_m}$ gives the corresponding matrix element for the transition. If we number the states in the basis and write $\ket{q}=\ket{n_1 n_2 \ldots n_m}$ and $\ket{p}=\ket{n_1' n_2' \ldots n_m'}$ as column vectors filled with zeros and a 1 for the $q$th or $p$th row, respectively, $U_{pq}=\bra{p}U\ket{q}$. $|U_{pq}|^2$ is the transition probability from $\ket{q}$ to $\ket{p}$ for the optical system under study. The total probability of finding a photon in an output state $\ket{n_1 n_2 \ldots n_m}$ can be interpreted as the Feynman sum of all the possible photon paths that end with the desired number of photons in each mode. 

Apart from this description of $\varphi$, we can write the elements of $U$ from the permanent of different submatrices of $S$\cite{Cai53,Sch04}.

\subsection{Effective Hamiltonians}
Photons are bosons and do not interact directly. Any Hamiltonian involving only photons must be indeed an effective Hamiltonian and all the changes come from the interaction with an intermediate material system. The molecules in the different media of the optical elements of a linear optics multiport can be modelled with great accuracy by a collection of two-level systems. If the photons are far from the resonant frequencies of each medium in the system, we can use adiabatic elimination to factor out any explicit coupling to the atoms of the optical elements\cite{Ste95}, much like we can describe a passive dielectric only by its index of refraction instead of speaking of multiple absorptions and reemissions. 

There are two related points of view when describing the evolution of photons in linear optics: the unitary evolution operators, $U$, and the Hamiltonians, $H$.

The evolution of a quantum state $\ket{\psi(t)}$ with time is the solution of the Schr\"odinger equation
\begin{equation}
i\hslash \frac{d\ket{\psi(t)}}{dt}=H \ket{\psi(t)},
\end{equation}
with a Hamiltonian $H$ and the reduced Planck's constant $\hslash$. The initial state evolves according to a unitary operator as
\begin{equation}
\ket{\psi(t)}=U(t) \ket{\psi(0)}
\end{equation}
with $U(t)=e^{-\frac{it H}{\hslash}}$. 

We can absorb $-\hslash$ into the Hamiltonian. Similarly, we can do away with time. Light crosses the whole optical network and partial evolution is generally not interesting except for some cases in which the fractional depth of a photon into a uniform optical system can play a role equivalent to $t$. In our finite-dimensional description with a fixed number of photons, $n$, the unitary evolution matrix $U=e^{iH_U}$ is the matrix exponential of $iH_U$ for a Hermitian matrix $H_U$ which gives the effective Hamiltonian. 

In the following, we will directly speak only of these matrices. While the Hermitian matrices we find can be interpreted as effective Hamiltonians, we can think solely in terms of group theory with a description of the transformations in the Lie groups of the quantum evolution matrices and the corresponding transformations in the associated Lie algebras.

The multiphoton Hamiltonian can be deduced from different points of view. Using a Heisenberg picture approach, we can envelope the creation and annihilation operator vectors with matrices derived from $S$ and write the effective Hamiltonian for any known optical network made from beam splitters, phase shifters and parametric amplifiers\cite{LN04}. The multiple photon Hamiltonian for linear optical networks can also be deduced from the scattering matrix by working with a coherent state representation\cite{FX97}. Alternatively, the decomposition of the single photon Hamiltonian $H_S$ in terms of the basis matrices of the corresponding algebra gives us a way to compute $H_U$ using the Jordan-Schwinger map\cite{AC05,ALN06}.

In this paper, we present an approach which, while resting on an analysis in the Heisenberg picture, describes the effective Hamiltonian for $n$-photon states in a way which reminds of the Schr\"ondinger's picture. A description in terms of the Lie algebra (Hamiltonian) behind the unitary evolution gives us a simple way to study linear optical systems from the differential of the homomorphism.

\section{The Hamiltonian for multiple photons}
We work with unitary matrices $S$ and $U$ and Hermitian matrices $H_S$ and $H_U$ so that $S=e^{iH_S}$ and $U=e^{iH_U}$. We call $S$ to the unitary $m\times m$ matrix that gives the evolution for a single photon state, which can be identified with the classical scattering matrix of the linear optical system. We call $U$ to the matrix describing the evolution for a general state of $n$ photons.

In terms of group theory, the unitary matrices belong to the unitary group of dimension $m$ and $M$, with $S\in U(m)$ and $U\in U(M)$, and the Hermitian matrices, up to a factor $i$, belong to the corresponding associated unitary Lie algebra with $iH_S \in \mathfrak{u}(m) $ and $iH_U \in \mathfrak{u}(M) $. All the results we present can be directly adapted to an alternative description of quantum optics using the special unitary group and its algebra. The unitary matrices would have determinant 1 and the Hermitian matrices would be traceless. The results would hold up to an unmeasurable global phase $e^{i\Phi}$ in $S$ that becomes $e^{i n\Phi}$ in $U$.

We consider $U(m)$ and $U(M)$ as (compact) Lie groups. This means they are compact manifolds and their Lie algebras $\mathfrak{u}(m)$~resp.~$\mathfrak{u}(M)$ are the tangent spaces to $U(m)$~resp.~$U(M)$ at the identity $I_m$~resp.~$I_M$. Moreover, the exponential map $\mathfrak{u}(M) \to U(M)$ is well-defined and surjective (cf.\cite[Sect.~18.4]{Gal11}).

We want to write the algebra homomorphism that allows us to write $H_U$ in terms of the elements of $H_S$. First we need to show:

\begin{lemma}\label{lemma:1}
The photonic homomorphism $\varphi$ is $C^\infty$.
\end{lemma}

\begin{proof}
This is trivial, since the entries of $U=\varphi(S)$ are polynomial expressions in the entries of the matrix $S$ (see Equation (\ref{HeisenbergU})).
\end{proof}

This allows us to prove:

\begin{theorem}\label{prop:1}
Let $\varphi: U(m)\to U(M)$ be the photonic homomorphism and consider the differential map $d\varphi: \mathfrak{u}(m) \to \mathfrak{u}(M)$. The diagram 
\[
\begin{tikzcd}
\mathfrak{u}(m) \arrow{d}[swap]{exp} \arrow{r}{d\varphi} &\mathfrak{u}(M) \arrow{d}{exp} \\ 
 U(m) \arrow{r}[swap]{\varphi} & U(M) \end{tikzcd}
\]
is commutative, i.e., $\varphi(\exp{(X)}) = \exp{(d\varphi(X))}$ for every $X \in \mathfrak{u}(m)$.
\end{theorem}

\begin{proof}
This follows from Warner\cite{War83} together with Lemma \ref{lemma:1}.
\end{proof}

We can now express the differential in terms of creation-annihilation operators as follows:

\begin{theorem}\label{prop:35}
Let $d\varphi:\mathfrak{u}(m)\to \mathfrak{u}(M)$ and $H_S=(H_{Sij}) \in \mathfrak{u}(m)$, $\ket{q}=\ket{n_1n_2\ldots n_m}$ and $\ket{p}=\ket{n_1'n_2'\ldots n_m'}$, then
\begin{equation}
\label{algebraHomomorphism}
iH_{Upq}=d\varphi(iH_{S})_{pq}=\bra{p}i\sum_{j=1}^m\sum_{l=1}^m  H_{Sjl}\hat{a}_j^\dag \hat{a}_l\ket{q}
\end{equation}
for $\hat{a}_j^\dag$ resp.~$\hat{a}_l$ the creation resp. annihilation operator in the $j$-th resp. $l$-th mode.
\end{theorem}
\begin{proof}
We work with the differential map and the basis states to find
\begin{equation}
iH_U\ket{n_1n_2\ldots n_m}=d\varphi(iH_S)\ket{n_1n_2\ldots n_m}.
\end{equation}
The elements of $H_U$ are deduced from the effects of $U=\varphi(S)=\varphi(e^{iH_S})$ on the basis states as shown in Equation (\ref{HeisenbergU}). We take the elements of $S$ in terms of the exponential of $iH_S$ and work around the identity taking $S=e^{iH_St}$ so that $S=I_m$ and $U=I_M$ for $t=0$. Then
\begin{align}
\label{dphi}
&iH_U\ket{n_1n_2\ldots n_m}=\left.\frac{d}{dt}\varphi\left(e^{iH_St}\right)\ket{n_1n_2\ldots n_m}\right\vert_{t=0}\nonumber\\
=&\left.\frac{d}{dt}\prod_{k=1}^{m} \frac{\left(\sum_{j=1}^{m} e^{iH_St}_{jk}\hat{a}_j^\dag\right)^{n_k}}{\sqrt{n_k!}}\ket{0\ldots 0}\right\vert_{t=0}\nonumber\\
=&\!\sum_{l=1}^{m}\!\frac{1}{\sqrt{n_l!}}\frac{d}{dt}\!\left.\left (\sum_{j=1}^{m} e^{iH_St}_{jl} \hat{a}^{\dag}_l\right )^{\!n_l}\right\vert_{t=0}  \prod_{k\neq l} \frac{\left(\sum_{j=1}^{m}\delta_{jk}\hat{a}_j^{\dag}\right)^{n_k}}{\sqrt{n_k!}}\ket{0\ldots 0}\nonumber\\
=&\sum_{l=1}^{m}\left( \sqrt{n_l} \sum_{j=1}^{m} iH_{Sjl} \hat{a}^{\dag}_j\right ) \frac{\hat{a}^{\dag n_l-1}_l}{\sqrt{(n_l-1)!}}\prod_{k\neq l}\frac{\hat{a}_k^{\dag n_k}}{\sqrt{n_k!}}\ket{0\ldots 0}\nonumber\\
=&\sum_{l=1}^{m}\sqrt{n_l} \sum_{j=1}^{m} iH_{Sjl} \hat{a}^{\dag}_j \ket{n_1 n_2\ldots n_l-1 \ldots n_m}\nonumber\\
=&\sum_{l=1}^{m}\sum_{j=1}^{m} iH_{Sjl} \hat{a}^{\dag}_j  \hat{a}_l  \ket{n_1 n_2\ldots n_l \ldots n_m}.
\end{align}

If we number the basis states, with $\ket{q}=\ket{n_1n_2\ldots n_m}$ and $\ket{p}=\ket{n_1'n_2'\ldots n_m'}$, the elements of $iH_U$ are
\begin{equation}
\label{HUpq}
\bra{p}iH_U\ket{q}=\bra{p}\sum_{l=1}^{m}\sum_{j=1}^{m} iH_{Sjl} \hat{a}^{\dag}_j  \hat{a}_l \ket{q}.
\end{equation}
\end{proof}

The evolution can be written as a sum of terms involving a single photon changing its mode, including changes from one mode to itself, where we have the photon number operator $\hat{n}_l=\hat{a}^{\dag}_l\hat{a}_l$. We can interpret the evolution in terms of single photon processes from the weighted sums
\begin{equation}
iH_{Uqq}=\sum_{l=1}^{m}i n_l H_{Sll}
\end{equation}
for the diagonal and
\begin{equation}
iH_{Upq}=\sum_{l=1}^{m}\sum_{j\neq l} i \sqrt{(n_j\!+\!1) n_l}H_{Sjl}\langle p | n_1 n_2\ldots n_j\!+\!1 \ldots n_l\!-\!1 \ldots n_m\rangle
\end{equation}
when $p\neq q$. The terms of $iH_U$ coming from a different state, $\overlap{p}{q}=0$, can only include one element of $H_S$. There is a contribution only if the input state $\ket{n_1n_2\dots n_m}$ is ``one photon away'' from the output state $\ket{n_1'n_2'\ldots n_m'}$. If $\overlap{p}{q}\neq 0$, we keep all the terms where the photons move to their original mode and there is one photon number operator $\hat{n}_l$ per occupied mode (a sum of $H_{Sll}$ terms multiplied by the corresponding $n_l$ factor).

In the starting effective Hamiltonian $H_S$ we have a Hermitian matrix with complex conjugate terms at indices $jl$ and $lj$, $H_{Sjl}=H_{Slj}^*$, so that in Equation (\ref{HUpq}) we can rewrite the sums as
\begin{equation}
\sum_{l=1}^{m}\sum_{j=1}^{m} H_{Sjl} \hat{a}^{\dag}_j  \hat{a}_l = \frac{1}{2}\sum_{l=1}^{m}\hat{n}_l H_{Sll}+
\sum_{l=1}^{m}\sum_{j< l }^{m} H_{Sjl} \hat{a}^{\dag}_j  \hat{a}_l +c.c.
\end{equation}

As expected, the resulting operator commutes with the total photon number operator
\begin{equation}
\hat{n}=\sum_{k=1}^{m}\hat{a}^\dag_k\hat{a}_k
\end{equation}
so that
\begin{equation}
[ H_U, \hat{n}]=0,
\end{equation}
showing the number of photons is preserved in the linear optical system. 

The description we have used is limited to linear systems that do not change the number of photons. For these systems we have a clear model and we know how to use passive elements to build a physical system implementing any desired $S$ matrix\cite{RZB94}. There are, however, linear optical operations that are not covered by this description. The most general possible linear operation is given by Bogoliubov transformations of the form  
\begin{equation}
\hat{a}_i \to \sum_{j}(u_{ij}\hat{a}_j+v_{ij}\hat{a}^{\dag}_j)
\end{equation}
which include phenomena like squeezing that give Hamiltonians with terms $\hat{a}_i\hat{a}_j$ or $\hat{a}^\dag_i\hat{a}^\dag_j$ that do no preseve the total photon number and require more sophisticated optical equipment\cite{Bra05}.

\section{Proportionality rules for two modes}
In this section, we give two simple examples in which working with Lie algebras gives an insight on the behaviour of linear optical systems with two input ports and an arbitrary number of photons. This kind of analysis shows the power of using group theory in the study of quantum optics. For two ports, the unitary evolution matrices belong to the $SU(2)$ group, which can be mapped to the rotation group in three dimensions. This has been used in the past to study photon evolution from angular-momentum transformations\cite{CST89} and has a particularly nice description in terms of the Jordan-Schwinger transformation\cite{AC05}.

In order to simplify the proofs in this section, we will fix the global phases and we will consider the homomorphism $\varphi: SU(m)\to SU(M)$. We restrict ourselves to the case of two input modes $\varphi: SU(2)\to SU(M)$ ($M=n+1$).
Since $d\varphi: \mathfrak{su}(2)\to \mathfrak{su}(M)$ is a Lie algebra isomorphism, the space $\mathfrak{h}:=d\varphi(\mathfrak{su}(2))$ is a  
subalgebra of $\mathfrak{su}(M)$. For $\mathfrak{su}(2)$  we will choose the following basis $\{i\sigma_x,i\sigma_y, i\sigma_z\}$  where $\sigma_x,\sigma_y,\sigma_z$ are the Pauli matrices, \emph{i.e.},
$$\begin{aligned}\sigma _{x}:={\begin{pmatrix}0&1\\1& 0\end{pmatrix}},\quad \sigma _{y}:={\begin{pmatrix}0&-i\\i&0\end{pmatrix}},\quad\sigma _{z}:={\begin{pmatrix}1&0\\0&-1\end{pmatrix}}\,.\end{aligned}$$
with the well-known commutation relations 
$$\begin{aligned}\left[\sigma _{x},\sigma _{y}\right]&=2i\sigma _{z},\quad \left[\sigma _{y},\sigma _{z}\right]=2i\sigma _{x},\quad \left[\sigma _{z},\sigma _{x}\right]=2i\sigma _{y}.\end{aligned}$$
Then, since $d\varphi$ is  a Lie algebra isomorphism, $iJ_x:=id\varphi(\sigma_x)$, $iJ_y:=id\varphi(\sigma_y)$, $iJ_z:=id\varphi(\sigma_z)$ becomes a basis of $\mathfrak{h}$ with the following commutation relations
$$\begin{aligned}\left[J_{x},J _{y}\right]&=2iJ_{z},\quad \left[J_{y},J _{z}\right]=2iJ_{x},\quad \left[J_{z},J_{x}\right]=2iJ_{y}.\end{aligned}$$
The expression of $J_z$ is really simple in terms of the photon number operators for each port with $$J_z=d\varphi(\sigma_z)=\hat{n}_1-\hat{n}_2.$$ This is the photon difference operator which appears in the description of homodyne measurement\cite{Bra90,VG93}. Instead of restricting to the usual case of a balanced beam splitter where one of the inputs is a photon number state and the second a local oscillator described by a coherent state, we give a general description in terms of photon numbers. This completes previous similar analyses from a different point of view\cite{CST89, Wal87}.

First, we look into the expected values for the photon difference for an input state $\ket{\alpha}=\ket{n-k,k}$ entering a linear two-port with a scattering matrix $S$, which, at the output, becomes $\ket{\beta}=\varphi(S)\ket{\alpha}$. The mean photon difference is always proportional to $n-2k$ with a proportionality constant $C(S)$ which depends only on $S$ and appears for any input state: 

\begin{theorem}
Given $\ket{\alpha}=\ket{n-k,k}$, for any $S\in SU(2)$, let us denote by $\ket{\beta}=\varphi(S)\ket{\alpha}$. Then there exist a constant $C(S)$ such that
$$
\bra{\beta}\left(\hat{n}_1-\hat{n}_2\right)\ket{\beta}=C(S)(n-2k).
$$
\end{theorem}
\begin{proof}
 Denote $U=\varphi(S)$. Then, $U$ is in the subgroup $\varphi(SU(2))\subset SU(M)$ when $n\geq 2$ (or we have the trivial identification $\varphi(SU(2))= SU(2)$ when there is just one photon, $n=1$). Then,
$$
 \bra{\beta}iJ_z\ket{\beta}=\bra{\alpha}U^\dag iJ_z U\ket{\alpha}.
$$
Note that $U^\dag iJ_z U$ is the image of $iJ_z$ by the adjoint map ${\rm Ad}_{U^{\dag}}(iJ_z)$ (see Warner\cite{War83}, chapter 3). But since  $U$ belongs to the subgroup $\varphi(SU(2))$, the map ${\rm Ad}_{U^\dag}:\mathfrak{h}\to \mathfrak{h}$ is an automorphism.  Hence, because $\{J_x,J_y,J_z\}$ is a basis of $\mathfrak{h}$, there should exist three real numbers $a,b,c$ such that
$$
U^\dag iJ_z U=aiJ_x+biJ_y+ci J_z.
$$
Moreover, 
$$
\begin{aligned}
\bra{\alpha}aJ_x+bJ_y+c J_z\ket{\alpha}=&(n-2k)c,
\end{aligned}
$$
where we have used that $\bra{\alpha}J_x\ket{\alpha}=\frac{-i}{2}\bra{\alpha}[J_z,J_y]\ket{\alpha}=\frac{-i}{2}\bra{\alpha}\left(J_zJ_y-J_yJ_z\right)\ket{\alpha}=0$ because $\ket{\alpha}$ is an eigenstate of $J_z$, and likewise $\bra{\alpha}J_y\ket{\alpha}=0$.
Therefore,
$$
\bra{\beta} J_z \ket{\beta}=(n-2k)c
$$
and the theorem follows.
\end{proof}

The above theorem allow us to state the following rule of proportionality:
\begin{corollary}
 Let $\ket{\alpha_1}=\ket{n-k_1,k_1}$, $\ket{\alpha_2}=\ket{n-k_2,k_2}$, assume  that $k_2\neq n/2$, denote by $\ket{\beta_1}=\varphi(S)\ket{\alpha_1}$, $\ket{\beta_2}=\varphi(S)\ket{\alpha_2}$. Then for any $S\in SU(2)$,
 $$
 \bra{\beta_1}\left(\hat{n}_1-\hat{n}_2\right)\ket{\beta_1}=\frac{n-2k_1}{n-2k_2}\bra{\beta_2}\left(\hat{n}_1-\hat{n}_2\right)\ket{\beta_2}.
 $$
\end{corollary}
Observe that the condition $k_2\neq n/2$ is for free when $n$ is odd.  In the particular case when $n$ is even (by setting $k_1= n/2$ and $k_2\neq n/2$ in the above corollary), we can state the following equipartition rule for the $\ket{n/2,n/2}$ state:

\begin{corollary}\label{cor:3}
Let $k$ be a nonnegative integer, then for any $S\in SU(2)$, the state $\ket{k,k}$ evolves to $\ket{\beta}:=\varphi(S) \ket{k,k}$ in such a way that
 $$
 \bra{\beta}\hat{n}_1\ket{\beta}=\bra{\beta}\hat{n}_2\ket{\beta}=k.
 $$
\end{corollary}
This means that, if we evenly distribute $n$ photons into the two possible input ports, the expected photon number at each output is always $n/2$ for any input port. This result agrees with previous analyses with Jacobi polynomials\cite{CST89} and it is in line with our intuition that a two-port is basically some kind of beam splitter with different coupling ratios which just redistributes the inputs. For equal photon numbers, the terms that are added and subtracted from each input cancel on average. 

By the properties of the Lie algebra $\mathfrak{su}(m)$, Corollary \ref{cor:3} can be extended to an arbitrary number of ports:

\begin{theorem}
Let $k$ be a nonnegative integer, then for any $S\in SU(m)$, the state $\ket{k,k, \ldots , k}$ evolves to $\ket{\beta}:=\varphi(S) \ket{k,k, \ldots , k}$ in such a way that
 $$
 \bra{\beta}\hat{n}_1\ket{\beta}=\bra{\beta}\hat{n}_2\ket{\beta}=\cdots =\bra{\beta}\hat{n}_m\ket{\beta} =k.
 $$
\end{theorem}

We will only sketch the proof and leave the details for the reader. The procedure is similar to the case with two ports. There are $m^2-1$ generalized Gell-Mann matrices, grouped in three matrix families, which generate the $\mathfrak{su}(m)$ algebra for an arbitrary $m$ \cite{BK08}. The photon number difference operator between any two ports is $\hat{n}_j-\hat{n}_l=-i d\varphi(D_{jl})$ for a diagonal matrix filled with zeroes except for the elements $D_{jj}=i$ and $D_{ll}=-i$. Now, for an input state with $k$ photons in each input, 
\begin{equation}
-i \bra{k,\ldots,k}U^\dag \hat{n}_j-\hat{n}_l U \ket{k,\ldots,k} 
\end{equation}
can be written as
\begin{equation}
-i \bra{k,\ldots,k}  {\rm Ad}_{S^{\dag}}(D_{jl}) \ket{k,\ldots,k}.
\end{equation}

The adjoint map can be computed from the known commutation relations of $D_{jl}$ with the generalized Gell-Mann matrices to show the average photon number difference must be 0 for any pair $j$ and $l$. 

This result agrees well with our intuition that, if a linear optics multiport can be described a concatenation of two port beam splitters and phase shifters, which do not affect the photon number average, then we have a series of steps for which the mean photon number does not change. In all the two port splitters we just redistribute the photons and the combined effect will not modify the final average. 

\section{Example for two photons in two modes}
We can see a simple example of our result for a system with two photons in two modes ($n=m=2$). For two modes, we define a scattering matrix
\begin{equation}
S=\left( \begin{array}{cc}
S_{11} & S_{12}  \\
S_{21} & S_{22} \end{array} \right).
\end{equation}
and a Hamiltonian
\begin{equation}
H_S=\left( \begin{array}{cc}
H_{S11} & H_{S12}  \\
H_{S21} & H_{S22} \end{array} \right).
\end{equation}

If we label the available photon states as $\ket{1}=\ket{20}$, $\ket{2}=\ket{02}$ and $\ket{3}=\ket{11}$, using Equation ($\ref{HeisenbergU}$) we obtain the unitary evolution matrix 
\begin{equation}
\label{U22}
U=\left( \begin{array}{ccc}
S_{11}^2 & S_{12}^2 & \sqrt{2}S_{11}S_{12} \\
S_{21}^2 & S_{22}^2 & \sqrt{2}S_{21}S_{22} \\
\sqrt{2}S_{11}S_{21} & \sqrt{2}S_{12}S_{22} & S_{11}S_{22}+S_{12}S_{21} \end{array} \right).
\end{equation}

From Equation (\ref{HUpq}), we can give the Hamiltonian $H_U$ in terms of the elements of $H_S$ as
\begin{equation}
\label{HU22}	
H_U=\left( \begin{array}{ccc}
2H_{S11} & 0 & \sqrt{2}H_{S12} \\
0 & 2H_{S22} & \sqrt{2}H_{S21} \\
\sqrt{2}H_{S21} & \sqrt{2}H_{S12} & H_{S11}+H_{S22}\end{array} \right).
\end{equation}

We can check the results for the simple example of the evolution of two photons inside a balanced beam splitter. The scattering matrix, up to a global phase, is
\begin{equation}
S=\frac{1}{\sqrt{2}}\left( \begin{array}{cc}
1 & 1 \\
1 & -1 \end{array} \right).
\end{equation}
We can find the corresponding Hamiltonian
\begin{equation}
H_S=\left( \begin{array}{cc}
0.46008 & -1.11072  \\
-1.11072 & 2.68152 \end{array} \right)
\end{equation}
either from the results in\cite{LN04} or by computing $H_S=-i \ln (S)$.

The Hamiltonian $H_U$, substituting in Equation (\ref{HU22}), is
\begin{equation}
H_U=\left( \begin{array}{ccc}
0.92016  & 0 & -1.57080 \\
0 & 5.36304 & -1.57080 \\
-1.57080 & -1.57080  & 3.14160 \end{array} \right)
\end{equation}
and we can check the result is correct by computing $e^{iH_U}$ and seeing it is, indeed, the unitary matrix
\begin{equation}
U=\left( \begin{array}{ccc}
\frac{1}{2} & \frac{1}{2} &\frac{1}{\sqrt{2}}\\
\frac{1}{2} & \frac{1}{2} &-\frac{1}{\sqrt{2}} \\
\frac{1}{\sqrt{2}} & -\frac{1}{\sqrt{2}} & 0 \end{array} \right)
\end{equation}
we expected from Equation (\ref{U22}). 

Observe that the expected values of $\hat{n}_1$ and $\hat{n}_2$ for the state $U\ket{1,1}=\frac{1}{\sqrt{2}}\ket{2 0}-\frac{1}{\sqrt{2}}\ket{0 2}$ are $1$, in agreement with our equipartition rule for the $\ket{1,1}$ state. 
Moreover, the expectation value of $\hat{n}_1-\hat{n}_2$ for the states $U\ket{2,0}$ and $U\ket{0,2}$ is $0$ in agreement, as well, with our proportionality rule ($0=\frac{-2}{2}0$).

\section{Discussion}
We have presented an alternative derivation of the formula for the effective Hamiltonian determining the evolution of $n$ photons through an $m$-mode linear optics multiport that preserves the photon number based on the differential form of the unitary evolution. This description has a reduced number of degrees of freedom, which makes studying the photon evolution easier. We have $m^2$ real parameters from $H_S$ instead of the $M^2$ parameters of a general $M\times M$ Hermitian matrix. The $H_U$ matrix has multiple null entries and, as it is Hermitian, we only need to compute explicitly the upper or lower triangular matrix plus the diagonal. Finding the unitary evolution matrix $U$ still presents some computational challenges. In particular, computing the matrix exponential can be a bottleneck.

Apart from the computational implications, expressing the evolution as an algebra homomorphism can be useful to study linear optical networks. Some quantum optics problems might be easier to tackle with a description in the Lie algebra using the wealth of results from group theory, as shown from the proportionality rules we have described for linear devices with two input ports.

The presented result also gives a natural way to study optical systems that combine linear and nonlinear optical parts. Systems which include squeezing or parametric processes are usually described in terms of their Hamiltonians, which can be readily combined with the Hamiltonian of the linear part. In that regard, our analysis can also be extended to general linear optics networks where the number of photons is not conserved like in parametric amplifiers\cite{Leo10}. These systems still admit a linear description with the Lie group of quasi-unitary matrices and its corresponding associated algebra\cite{LN04}. 

\begin{acknowledgements}
This work has been funded by: Spanish Ministerio de Econom\'ia y Competitividad, Project TEC2015-69665-R, MINECO/FEDER, UE and Junta de Castilla y Le\'on VA089U16 (J.C. Garcia-Escartin); DGI-MINECO grant (FEDER) MTM2017-84851-C2-2-P and Universitat Jaume I, grant P1-B2016-07 (V. Gimeno); Ministerio de Econom\'ia y Competitividad (MINECO), grant MTM2015-65764-C3-2-P, and Universitat Jaume I, grant P1-1B2015-02 (J.J. Moyano-Fern\'andez).
\end{acknowledgements}
\bigskip
\noindent
\newcommand{\noopsort}[1]{} \newcommand{\printfirst}[2]{#1}
  \newcommand{\singleletter}[1]{#1} \newcommand{\switchargs}[2]{#2#1}

\end{document}